\documentclass[12pt]{article}

\usepackage{amsfonts,amssymb,amsmath,latexsym,epsfig,amsthm,graphicx,euscript,cite}

\newcommand{\pd}{\partial}

\DeclareMathOperator{\diag}{diag}

\def\kk{k}
\newcommand{\Dd}{\EuScript{D}^2}
\newcommand{\Dh}{\Dd_\mathrm{H}}
\newcommand{\DNh}{\tilde{\EuScript{D}}^2_\mathrm{H}}

\newtheorem{lemma}{Lemma}
\newtheorem{claim}{Claim}

\title{
Dynamics with Infinitely Many Time Derivatives in
Friedmann-Robertson-Walker background and Rolling Tachyons}

\author{Liudmila Joukovskaya\thanks{E-mail: \texttt{l.joukovskaya@damtp.cam.ac.uk}}\\
Centre for Theoretical Cosmology,\\ DAMTP, CMS, University of Cambridge,\\
Wilberforce Road, Cambridge CB3 0WA, UK}

\begin{document}

\maketitle

\date{}

\begin{abstract}
Open string field theory in the level truncation approximation is considered.
It is shown that the energy conservation law determines existence of rolling
tachyon solution. The coupling of the open string field theory action to a
Friedmann-Robertson-Walker metric is considered which leads to a new time
dependent rolling tachyon solution is presented and
possible cosmological consequences are discussed.
\end{abstract}

\section{Introduction}

Consideration of fundamental theories such as M/String Theory
in the cosmological context continues to attract attention in the literature.
One of the interesting questions is the role of tachyon in
String Theory and Cosmology. The great progress in our understanding
of tachyon condensation was made in the past
decade\cite{Sen2002,Sen2004,Gibbons}, but a lot of interesting issues
are still open. Among the most important ones is a better
understanding of the dynamics of tachyon condensation process.

In this context many works have been devoted to the study of
time-dependent solutions. Probably one of the most fascinating
frameworks for this is Open String Field Theory (OSFT) \cite{Witten} which
reattracted a lot of attention after a recent work \cite{Schn-2005}.
Despite the recent renewal of interest to OSFT no
smooth solutions interpolating between two inequivalent vacua
even at the lowest level truncation order \cite{Kost-Sam, MSZ} were found.
One of the exceptional features of the level truncation approach is
that corresponding action contains infinitely many time derivatives,
i.e. it is non-local.
Resulting models have a rich set of properties that might be essential for
the development of stringy cosmology. Substantial investigation
of this type of models was performed
\cite{MZ,Yang,AJK,Y,LJYV,FGN,LJ-TMF,VS-YV,VS,Prokh,BarnKam,AJ,
Arefeva,Calcagni,AJV-JHEP,Lulya-paper,
LJ-twofields,CMN,p-adic-infl,Lidsey,AJV-JHEP-2,diff-infl,NM} after
seminal paper \cite{MZ} where many useful issues were revealed,
among which the presence of rolling solution with widely
increasing oscillations. Further research of time-dependent rolling
solutions was also performed in particular in \cite{FH-2003}. Recently several
time dependent rolling tachyon solutions in OSFT were found \cite{Sch-2007,KORZ},
which confirmed puzzling
behavior of the solutions found earlier \cite{MZ,FH-2003,FH-2004}.
In these works investigations were performed in usual space-time
coordinates. At the same time is it known that in the light-cone
gauge the theory becomes local in light-cone time \cite{KaKi}.
In fact the question of the identity of light-like and time-like
cases is under investigation. In the light-like case (with dilatonic
damping), the gradient flow forces tachyon to asymptote to the
true vacuum at late times \cite{HS}. In the time-like case, rolling
configuration between two non-equivalent vacua is forbidden by
energy conservation law in Minkowski space time if one considers
only tachyon excitation in the level truncation approximation
\cite{MZ,Lulya-paper-2}. Investigations in the last case though
did not take into account any effects of gravity, this makes such an
investigation inconsistent from cosmological point of view.

It turns out that if we consider the same action coupled to the
gravity in Friedmann-Robertson-Walker (FRW) metric the situation changes:
there appears a tachyon solution which tends to the true vacuum at late
times \cite{Lulya-paper-2}. It is interesting to note that because
dilaton appears from the same
string sector as graviton, inclusion of dilaton into the tachyon
action can qualitatively reproduce behavior of the tachyon in curved
space. Thus obtained solution has accordance with the results obtained
in light-like and time-like cases at least at the lowest level approximation.

In the present work we consider scalar field dynamics with infinitely many time derivatives
minimally coupled to the Minskowski and Friedmann-Robertson-Walker gravitational backgrounds.

The structure of the work is the following. In the first section we will
give a brief introduction and physical motivation.
In the second section the model which appears from Open
String Field Theory will be presented and problem of
existence of interpolating solutions between two inequivalent vacua
will be discussed. In the third section we will consider the model
minimally coupled to gravity and demonstrate an intriguing difference
compared to the case without gravity: the existence of desired solution.
Numerical techniques will be described in Sec. 4.
Results of numerical calculations will be presented in Sec. 5.
Finally we will summarize main results.

\section{The Model}

The action of bosonic cubic string field theory has the form
\begin{equation}
S=-\frac{1}{g_0^2}\int(\frac{1}{2}\Phi \cdot Q_B \Phi+\frac{1}{3} \Phi \cdot (\Phi \ast \Phi)),
\label{action-full}
\end{equation}
where $g_0$ is the open string coupling constant, $Q_B$ is BRST operator,
$\ast$ is noncommutative product and $\Phi$ is the open string field containing
component fields which correspond to all the states in string Fock space.

Considering only tachyon field $\phi(x)$  at the level (0,0)
the action (\ref{action}) becomes
\begin{equation}
S=\frac{1}{g_0^2}\int d^{26}x\left[
\frac{\alpha^{\prime}}{2}\phi(x)\square\phi(x)+\frac{1}{2}\phi^2(x)
-\frac{1}{3}K^3 \Phi^3(x)-\Lambda
\right],
\label{action}
\end{equation}
where  $\alpha^{\prime}$ is the Regge slope, $K=\frac{3
\sqrt{3}}{4}$, $\phi$ is a scalar field, $\Phi=e^{k\square_g}\phi$,
$k=\alpha^{\prime} \ln K$,
$\square=\frac1{\sqrt{-g}}\pd_{\mu}\sqrt{-g}g^{\mu\nu}\pd_{\nu}$ and
$\Lambda =\frac{1}{6} K^{-6}$ was added to the potential to set the
local minimum of the potential to zero according Sen's conjecture
\cite{Sen1999}. In what follows we will work in units where $g_0=1$.

 The action (\ref{action}) leads to equation of motion
\begin{equation}
(\alpha^{\prime} \square+1)e^{-2 k \square} \Phi=K^3 \Phi^2.
\label{eom}
\end{equation}
The Stress Tensor for our system is\footnote{Note that here and below integration over $\rho$
is understood as limit of the following regularization
$$
\int_0^1d \rho f(\rho)= \lim _{\epsilon_1\to +0} \lim _{\epsilon_2\to +0}
\int^{1-\epsilon_2}_{\epsilon_1} d \rho f(\rho).
$$}
\begin{equation}
T_{\alpha\beta}(x)=
-g_{\alpha\beta}
\left(
    \frac{1}{2}\phi^2
    -\frac{\alpha^{\prime}}{2}\partial_\mu\phi\partial^\mu\phi
    -\frac{1}{3}K^3\Phi^3-\Lambda
\right)
-\alpha^{\prime} \partial_\alpha\phi\partial_\beta\phi
\end{equation}
$$
-g_{\alpha\beta}\,\kk
\int_0^1 d\rho
\left[
    (e^{k\rho\square}K^3 \Phi^2)
    (\square e^{-\kk\rho\square}\Phi)
    +
    (\partial_{\mu}e^{\kk\rho\square}K^3 \Phi^2)
    (\partial^{\mu}e^{-\kk\rho\square}\Phi)
\right]
$$
$$
+2\kk\int_0^1d\rho~
(\partial_{\alpha}e^{\kk\rho\square}K^3 \Phi^2)
(\partial_{\beta}e^{-\kk\rho\square}\Phi).
$$
The energy is defined as $E(t)=T^{00}$
and pressure as $P(t)_i=-T_i^i$ (no summation) and for our system are
$$
\mathcal{E}=\mathcal{E}_k+\mathcal{E}_p+\Lambda
+\mathcal{E}_{nl1}+\mathcal{E}_{nl2},~~~
\mathcal{P}=\mathcal{E}_k-\mathcal{E}_p-\Lambda
-\mathcal{E}_{nl1}+\mathcal{E}_{nl2}
$$
where
$$
{\cal E}_{k}=\frac{\alpha^{\prime}}{2}(\pd \phi)^2,~~~
{\cal E}_{p}=-\frac{1}{2}\phi^2+\frac{K^3}{3}\Phi^3
$$
$$
{\cal E}_{nl1}=k\int_{0}^{1}d \rho \left(
e^{k\rho \square} K^3 \Phi^2 \right)
 \left(-\square e^{-k\rho\square}\Phi\right),
$$
$$
{\cal E}_{nl2}=-k\int_{0}^{1} d \rho \left(\partial
e^{k\rho \square}K^3 \Phi^2\right) \left(\partial e^{-k \rho \square}\Phi\right).
$$
In this paper we will be interested in  spatially homogeneous configurations
for which Beltrami-Laplace operator used above takes the form $\square_g=-\pd^2$
and nonlocal operator $e^{k \Box_g}$ becomes $e^{k \Box_g}=e^{-k \pd^2}$.

The symbol $e^{ \rho\partial^2 }\varphi$ comprehend as\footnote{It
is easier to use integral representation for the operator
$e^{ \rho\partial^2 }$ while considering Minkowski background,
in the FRW case though it becomes impossible to generalize such an approach
and we define the operator $e^{ \rho\partial^2 }$
in terms of solution of the boundary value problem for diffusion
equation, see \cite{Lulya-paper} for the details.}
\begin{equation}
e^{\rho\partial^2 }\Phi(t)=C_{\rho}[\Phi](t)
\label{int-repr}
\end{equation}
where
$$
C_{\rho}[\Phi](t)=\frac{1}{\sqrt{4 \pi \rho}}
\int_{-\infty}^{+\infty} e^{-\frac{(t-t^{\prime})^2}{4
\rho}}\Phi(t^{\prime})d t^{\prime}.
$$
Nonlocal terms ${\cal E}_{nl_1}$ and ${\cal E}_{nl_2}$  contain  $e^{-k
\rho\pd^2}$ which as it can be easily seen might lead to the growing
kernel in the integral representation of the non-local operator,
that is why we will try to avoid calculation of $e^{-k \rho\pd^2}$
and  will use the following representation for nonlocal energy terms
${\cal E}_{nl1}$ and ${\cal E}_{nl2}$ which are valid on the
equation of motion for the scalar field
$$
{\cal E}_{nl1}=k\int_{0}^{1}d \rho \left((-\alpha^{\prime} \pd^2+1)
e^{(2-\rho)k \pd^2}\Phi \right)
 \left(\pd^2 e^{k\rho\pd^2}\Phi\right),
$$
$$
{\cal E}_{nl2}=-k\int_{0}^{1} d \rho \left(\partial (-\alpha^{\prime} \pd^2+1)
e^{(2-\rho)k\pd^2}
\Phi\right) \left(\partial e^{k \rho \pd^2}\Phi\right).
$$

\subsection{Energy conservation}

Energy conservation laws always have the deep physical sense, note that because
of presence of infinitely many time derivatives we need to prove energy conservation explicitly.
Similar investigation was performed in \cite{MZ}, although in that work it was used
representation of pseudo-differential operator via  the summation over
infinite series expansion with which one always needs
to be very careful with regard to convergence issues.
The approach used here is advantageous from the point of view of
numerical calculations, because in order
to define action of the exponential operator we need to do only
one well defined integration.

Taking into account that models with different types of potentials
are currently under the consideration in the literature
\cite{Calcagni,diff-infl,Lulya-paper} we will show the energy
conservation for arbitrary potential.

\begin{claim}\footnote{For simplicity we will use symbolic notation for nonlocal
operator $e^{\pd^2}$ keeping in mind that it is in fact defined by (\ref{int-repr}),
also integration over $\rho$ must be understood as limit of the corresponding regularization
as indicated earlier.}
The Energy
$$
E=\frac{\alpha^{\prime}}{2}(\pd\phi)^2-\frac{1}{2}\phi^2+V(\Phi)+\Lambda+
\kk\int^1_0 d \rho~((-\alpha^{\prime} \pd^2 +1) e^{(2-\rho) \kk
\pd^2} \Phi) \overleftrightarrow{\partial}(\partial e^{\kk\rho\pd^2}
\Phi),
$$
is conserved on the solutions of equation of motion
$$
(-\alpha^{\prime} \partial^2+1)e^{2\kk\pd^2}\Phi=\frac{\pd
V(\Phi)}{\pd \Phi}
$$
of the corresponding actions
$$
S=\frac{1}{g_0^2}\int d^{26}x\left[
\frac{\alpha^\prime}{2}\phi(x)\square\phi(x)+\frac{1}{2}\phi^2(x)
-V(\Phi)-\Lambda \right],
$$
where $V(\Phi)$ is any polynomial potential
and $A\overleftrightarrow{\partial}B=A\pd B-B \pd A$.
\end{claim}
\begin{proof} We compute
$$
\frac{d E(t)}{dt}=\alpha^{\prime} \pd\phi\pd^2 \phi-\phi\pd
\phi+\frac{\pd V(\Phi)}{\pd \Phi} \pd \Phi+ k\int^1_0 d
\rho((-\alpha^{\prime} \pd^2 +1)e^{(2-\rho)k\pd^2}\Phi)
\overleftrightarrow{\partial^2} (\partial e^{k\rho\pd^2}\Phi)
$$
Using following identity (for details on its derivation see Appendix)
$$
\int\limits_0^1 d\rho (e^{ \rho
\partial^2 }\varphi ) \overleftrightarrow{\partial^2} (e^{ (1
- \rho)\partial^2 } \phi)= \varphi \overleftrightarrow {e^{\partial^2 }}\phi,
$$
equation of motion and definition of field $\Phi$, we have
$$
\frac{d E(t)}{dt}=\alpha^{\prime} \pd \phi \pd^2 \phi-\phi\pd
\phi+\frac{\pd V}{\pd \Phi} \pd \Phi+\partial
\Phi\overleftrightarrow{e^{k\pd^2}}
 (\alpha^{\prime}\pd^2-1)e^{k\pd^2}\Phi=
 $$
 $$=
 \pd \Phi\left[(-\alpha^{\prime} \partial^2+1)e^{2k \pd^2}\Phi+\frac{\pd V}{\pd \Phi}
\right]=0.
$$
\end{proof}
In the next paragraph we  consider one important physical consequence
of this energy conservation law.

\subsection{Existence of the Rolling Solution}
We already indicated that from physical perspective we are interested in solutions
interpolating between two inequivalent vacua. We start our consideration by
looking for stationary configurations $\Phi_0$.
Substituting it into equation of motion (\ref{eom}) we get $\Phi_0=K^3\Phi_0^2$,
which has two constant solutions: $\Phi_0=0$ and $\Phi_0=K^{-3}$.
We should thus be looking for solutions interpolating between those
stationary points. The following claim though shows that energy conservation
forbids existence of such solutions.
\begin{claim}\footnote{Similar claim for $p$-adic string model was proved
in \cite{MZ}, which rules out the possibility that tachyon may roll monotonically
down from one extremum reaching the tachyon vacuum.}
There do not exist continuous solutions of equation (\ref{eom}) which satisfy
boundary conditions
\begin{equation}
\label{bc}
\lim \Phi(t)=
\begin{cases}
0,&t\to\infty,\\
K^{-3},&t\to-\infty
\end{cases}
\end{equation}
or vice-versa (in terms $t \to -t$).
\end{claim}
\begin{proof}
Let us assume existence of such solution and calculate energy
at the extremum  points, we get  $E(\Phi=0)=\Lambda$
and $E(\Phi=K^{-3})=-\frac{1}{6}K^{-6}+\Lambda$, i.e. energy values at
$t \to +\infty$  and  $t \to -\infty$ are different what due to conservation
law rules out existence of solutions satisfying (\ref{bc}).
\end{proof}
As we can see energy conservation law plays crucial role in the existence of the time dependent
solutions in the level truncation approximation to OSFT.
The above statement could potentially be generalized to the case of full OSFT because for
the action with cubic interaction solution interpolating between maximum and minimum
in the effective potential has to interpolate between vacua with different energy.

\section{The Model Coupled to the Gravity}
\label{rolling-tachyon}

In this section we would like to consider tachyon dynamics in FRW
background what allows us to take into account gravity effects and
makes the research more consistent from cosmological point of view.
Consider the model
\begin{equation}
\label{action-gr}
S=\frac{1}{g_0^2} \int d^4x\sqrt{-g}\left(\frac{m_p^2}{2}R+
\frac{1}{2}\phi \square_g \phi
+\frac{1}{2}\phi^2-\frac{1}{3}K^3\Phi^3-\Lambda \right),
\end{equation}
here $m_p^2=g_0^2 M_{pl}^2$ and we will work in units where $\alpha^{\prime}=1$.
As a particular metric we will consider the FRW
$$
ds^2=-dt^2+a^2(t)(d x_1^2+dx_2^2+dx_3^2),
$$
for which the Beltrami-Laplace operator for spatially-homogeneous configurations
takes the form $\square_g=-\pd^2-3 H(t) \pd=-\Dh$.
Scalar field and Friedmann equations are
\begin{equation}
(-\Dh+1)e^{2k \Dh}\Phi = K^3 \Phi^2,
\label{eom-gr}
\end{equation}
\begin{equation}
3H^2=\frac{1}{m_p^2}~{\cal E},~~~
3H^2+2\dot H={}-\frac{1}{m_p^2}~{\cal P}.
\label{Fr-eq}
\end{equation}
Inclusion of the gravity considerably modifies the dynamics of the
system. In Minkowski background we have shown that energy
conservation law forbids dynamical interpolation between two
inequivalent vacua. In FRW background energy of the scalar field
alone does not conserve any more due to the Hubble term and as a result
the restrictions on existence of such solutions are removed.

To find boundary conditions for possible solutions let us consider
constant scalar field solution, $\Phi_0$.
In this case the scalar field equation (\ref{eom-gr}) becomes
\begin{equation}
\Phi_0 = K^3 \Phi_0^2,
\label{eom-gr-st}
\end{equation}
and first equation in (\ref{Fr-eq})
\begin{equation}
3H_0^2=\frac{1}{m_p^2}~{\cal E} (\Phi_0). \label{Fr-eq-st}
\end{equation}
Equation (\ref{eom-gr-st}) has two solutions: $\Phi_{0_1}=0$ and
$\Phi_{0_2}=K^{-3}$, substituting which into (\ref{Fr-eq-st}) we
obtain corresponding values for a Hubble parameter $H_{0_1}=(18
K^6)^{-1/2}$ and $H_{0_2}=0$. Note that from cosmological
perspective we are interested only in positive values for the Hubble
function, so we can expect rolling solutions with the following boundary conditions
\begin{equation}
\label{bc-gr}
\lim \Phi(t)=
\begin{cases}
0,&t\to\infty,\\
K^{-3},&t\to-\infty,
\end{cases}
~~~~~
\lim H(t)=
\begin{cases}
(18 K^6)^{-1/2},&t\to\infty,\\
0,&t\to-\infty,
\end{cases}
\end{equation}
or vice-versa (in terms of $t \to -t$).

\begin{figure}
\centering
\includegraphics[width=71mm]{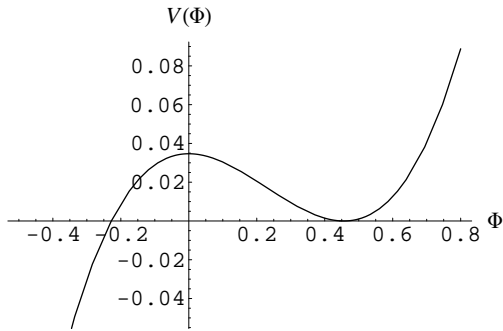}
\caption{The potential}
\label{cubic-potential}
\end{figure}

To analyze physical situation let us consider potential  in which
motion is expected. Naive extraction of potential from the model
action (\ref{action-gr}) results in
$V(\Phi)=-\frac{1}{2}\Phi^2+\frac{1}{3}K^3 \Phi^3+\Lambda$. The
constant  $\Lambda$ represents the $D$-brane tension which according
to Sen's conjecture must be added to cancel negative energy
appearing due to the presence of tachyon. We obtained two types
of solutions. The first one is an ordinary rolling solution which
starts from $\Phi=0$ and goes towards configuration $\Phi=K^{-3}$
which is associated with the true vacuum. This solution can be
interpreted as a description of the $D$-brane decay. The second one is
a rolling solution which goes in the opposite direction, which
appears in this model possibly because of the non-locality in the
interaction. It is known that nonlocal dynamics has many interesting
properties which are not possible in the local case. In particular
the ``slop effect'' \cite{MZ,FGN,AJ} which we can observe in the
obtained solutions (Fig.\ref{rolling-solution-ph-0const},
\ref{rolling-solution-ph-const0}) when the scalar field goes beyond
the values from which the scalar field configuration starts --
situation which is not possible in the local models. Potentially a
similar effect can initiate non-symmetry in the potential in
ekpyrotic \cite{Ekpyrotic} and cyclic cosmology \cite{Cyclic}.

\section{Numerical Solution Construction}

In order to construct numerical solution we operate with scalar field
equation of motion (\ref{eom-gr}) and the difference of equations
(\ref{Fr-eq}), specifically we solve
\begin{subequations}
\begin{equation}
\label{fr2-a}
(-\Dh+1)e^{2k \Dh}\Phi = K^3 \Phi^2,
\end{equation}
\begin{equation}
\label{fr2-b}
\dot{H}=-\frac{1}{2m_p^2}~(\cal P+\cal E).
\end{equation}
\end{subequations}
We discretize equations (\ref{fr2-a})-(\ref{fr2-b}) by introducing a lattice in $t$ variable and
then solving resulting system of nonlinear equations using iterative
relaxation solver controlling error tolerance with discrete $L_2$
and $L_{\infty}$ norms.

\subsection{Discretization}

When solving discretized equations (\ref{fr2-a})-(\ref{fr2-b}) the
nontrivial part from computational point of view is efficient
evaluation of $e^{2k\rho\Dh}\Phi$ for $\rho\in[0,2]$. This operator
could be interpreted in terms of initial value problem for the
following diffusion equation with boundary conditions
\cite{Lulya-paper}
\begin{equation}
\label{numdiff}
\partial_\rho\varphi(t,\rho)=
\partial^2_t\varphi(t,\rho)+3H(t)\partial_t\varphi(t,\rho),
\end{equation}
$$
\varphi(0,t)=\Phi(t),~~\varphi(\rho,\pm\infty)=\Phi(\pm\infty).
$$
Once solution of this equation is constructed we have $e^{2k\rho\Dh}\Phi(t)=\varphi(\rho,t)$.

To solve (\ref{numdiff}) we used second order Crank-Nicholson scheme which is based on
the following stepping procedure
$$
\varphi(t,\rho+\Delta_\rho)=
\left(1+k\Delta_\rho\DNh\right)\left(1-k\Delta_\rho\DNh\right)^{-1}\varphi+
o(\Delta_\rho^2\|\DNh\|),
$$
where $\DNh$ denotes discretization of $\Dh$ (it thus has a finite norm)
and $\Delta_\rho$ is a step size along $\rho$ variable.
We used the following tri-diagonal discretization scheme for $\Dh$
$$
\DNh=
\frac{1}{\Delta_t^2}
\begin{pmatrix}
-2     &      2 & 0      & \ldots & 0 \\
1      &     -2 & 1      & \ldots & 0 \\
\vdots & \ddots & \ddots & \ddots & \vdots \\
 0     & \ldots & 1      &-2      & 1 \\
 0     & \ldots &0       &-2      & 2
\end{pmatrix}
+
\frac{1}{2\Delta_t}\diag(H_{-N},\ldots,H_N)
\begin{pmatrix}
 0     &  0     & 0      & \ldots & 0\\
-1     &  0     & 1      & \ldots & 0\\
\vdots & \ddots & \ddots & \ddots & \vdots \\
 0     & \ldots &-1      & 0      & 1 \\
 0     & \ldots &0       & 0      & 0
\end{pmatrix},
$$
where $\Delta_t$ is the step size along $t$ and discretization of Hubble function is
$H_k=H(k\Delta_t)$.
This is a usual symmetric discretization scheme on the uniform lattice modified
on the interface to guarantee boundary condition (\ref{fr2-b}) for smooth solutions
which tend to constants as $t\to\pm\infty$.

\subsection{Comparison the different methods}

In order to exclude possible artifacts of the specific numerical
scheme described above we tried Chebyshev-pseudospectral method which
is known to generally have exponential convergence \cite{Forn}.
Such scheme is known to have very different properties \cite{Forn} compared to
finite difference scheme described above, but it produced the same
results up to the approximation error which provides confidence in
the existence of the rolling solutions reported in this work.

\section{Rolling Tachyon Solution}
\label{Rolling-Tachyon-solution}

\begin{figure}
\centering
\includegraphics[width=57mm]{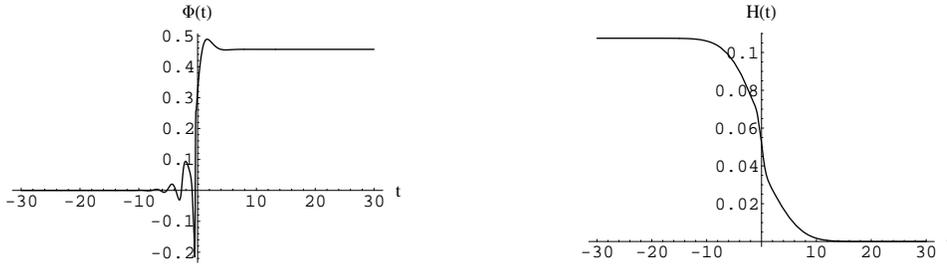}~~~~~~~~~~~~
\includegraphics[width=57mm]{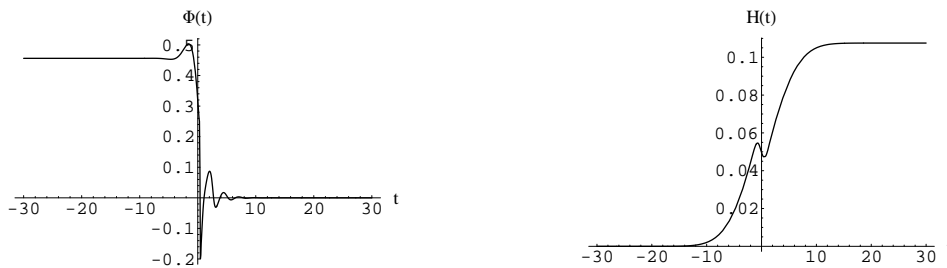}
\caption{Solutions of the scalar field (\ref{eom-gr}) and  Friedmann
equation (\ref{Fr-eq}) $\Phi$ and $H$ (left to right) for
$m_p^2=1$.} \label{rolling-solution-ph-0const}
\end{figure}

\begin{figure}[b]
\centering
\includegraphics[width=57mm]{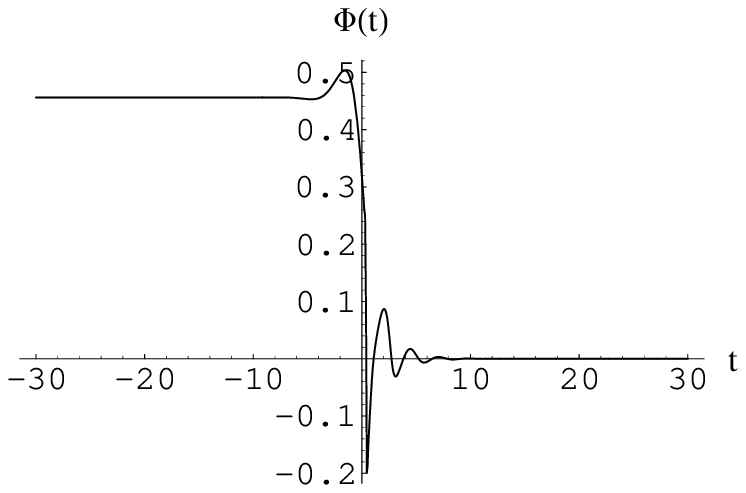}~~~~~~~~~~~~
\includegraphics[width=57mm]{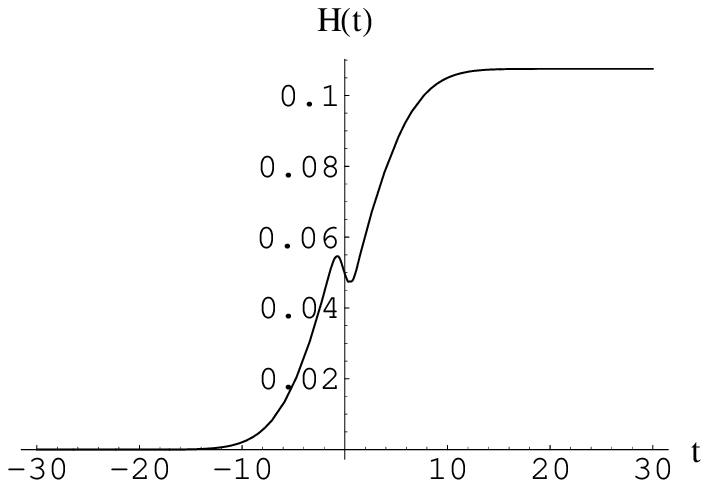}
\caption{Solutions of the scalar field (\ref{eom-gr}) and  Friedmann
equation (\ref{Fr-eq}) $\Phi$ and $H$ (left to right) for
$m_p^2=1$.} \label{rolling-solution-ph-const0}
\end{figure}

Solutions of (\ref{eom-gr}) and (\ref{Fr-eq}) are presented on Fig.\ref{rolling-solution-ph-0const}, \ref{rolling-solution-ph-const0}.
It is interesting to note that on these figures shapes of scalar field
look very similar up to reflection over vertical axis while shapes of
Hubble function are different. In order to ensure that this is not an
artefact or error in numerical calculations we performed the following test.
First, we performed reflection, $t\to-t$, of the scalar field from
Fig.\ref{rolling-solution-ph-0const} and compared in with scalar field from
Fig.\ref{rolling-solution-ph-const0}, while difference was small it was well
above the error tolerance. We also numerically computed discrepancies between
left and right sides of equations of motion where Hubble function was
from the solution while scalar field was taken from another solution and mirrored.
While the discrepancy for scalar field was small, the discrepancy for Hubble
function was very large, larger than the Hubble function itself.
This test makes us confident that qualitatively different shapes of
Hubble function presented on
Fig.\ref{rolling-solution-ph-0const}, \ref{rolling-solution-ph-const0} are correct.

Another interesting property of the solution on Fig.\ref{rolling-solution-ph-const0}
is that while scalar filed generally rolls down at first it climbs up, we
already noted this ``slop effect'' in the end of section \ref{rolling-tachyon}.
In fact similar behavior is typical for solutions of nonlocal equations.
In Minkowski metric such behavior was noted by many authors who used different
numerical techniques \cite{FGN,AJ,NM}. It is interesting that in FRW metric
this property is preserved.

\subsection{Solutions for different $m_p$}

Motivated by the fact that string scale does not exactly coincide
with Planck mass and as a consequence there is some freedom in settling $m_p^2$
we investigated how the shape of solutions behave for different $m_p^2$.
Solutions of (\ref{eom-gr}), (\ref{Fr-eq}) for different values of $m_p^2$
are presented on Fig.\ref{rolling-solution-ph-0const-diffmp} and
\ref{rolling-solution-ph-const0-diffmp}.

We can see that the profiles for different $m_p^2$ for the first
case  are very similar. It is interesting to note how on Fig.\ref{rolling-solution-ph-const0-diffmp} with the growth of $m_p^2$
oscillation in the profile of Hubble function disappears.

\begin{figure}
\centering
\includegraphics[width=51mm]{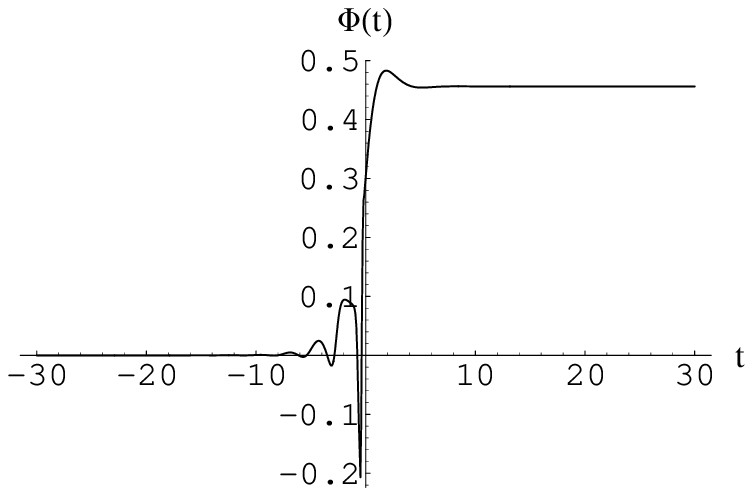}
\includegraphics[width=51mm]{f1.eps}
\includegraphics[width=51mm]{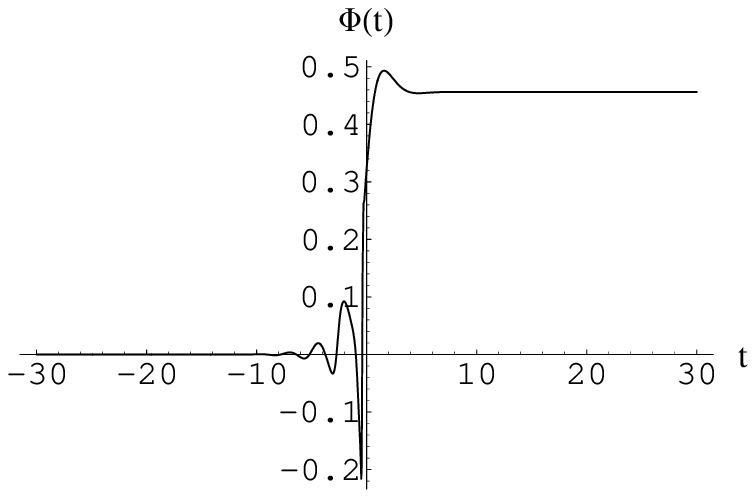}
\includegraphics[width=51mm]{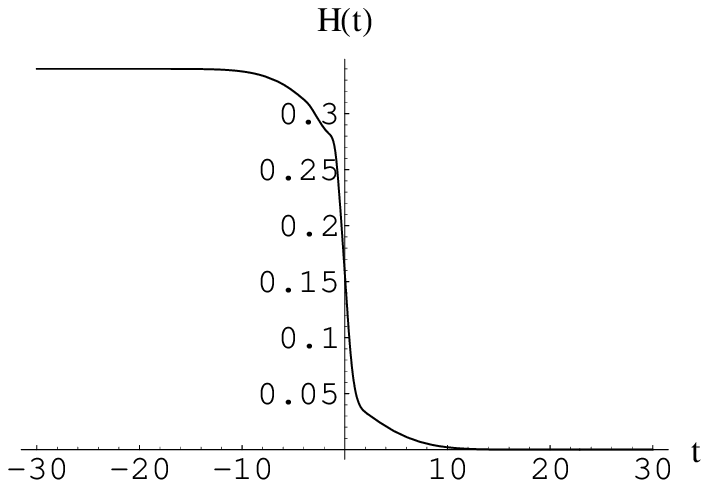}
\includegraphics[width=51mm]{H1.eps}
\includegraphics[width=51mm]{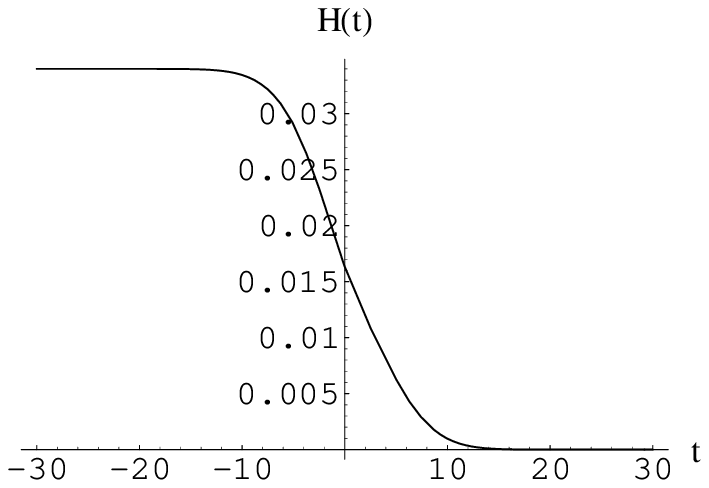}
\caption{Solutions of the scalar field (\ref{eom-gr}) and  Friedmann
equation (\ref{Fr-eq}) $\Phi$, $H$ for $m_p^2=0.1, 1, 10$ (left to right).
} \label{rolling-solution-ph-0const-diffmp}
\end{figure}

\begin{figure}
\centering
\includegraphics[width=51mm]{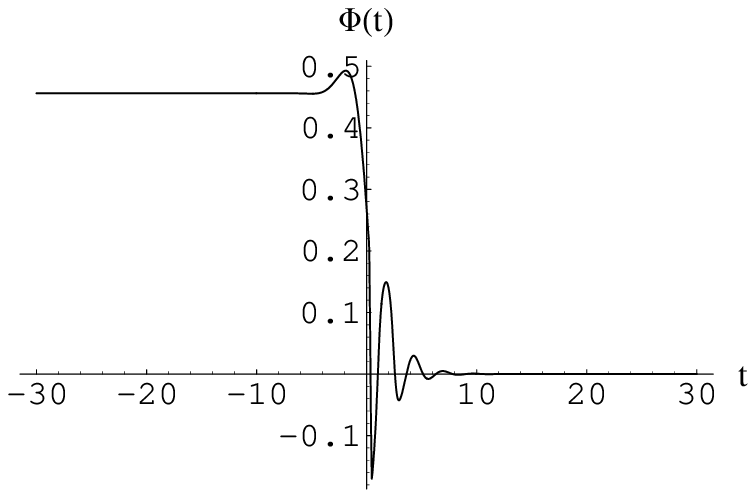}
\includegraphics[width=51mm]{f2.eps}
\includegraphics[width=51mm]{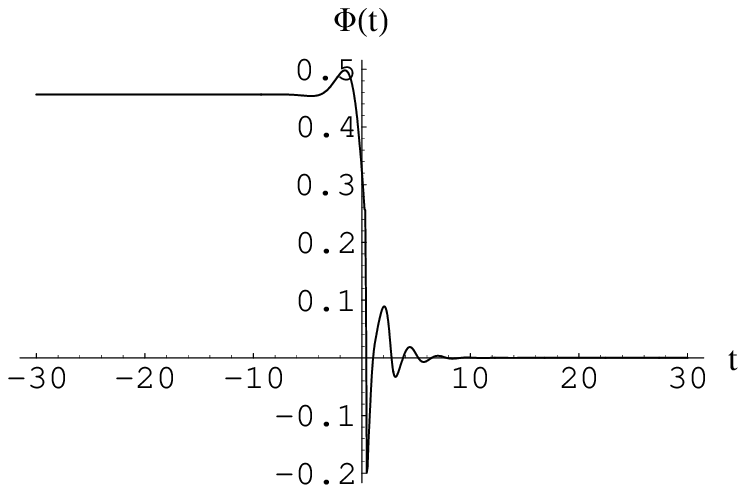}
\includegraphics[width=51mm]{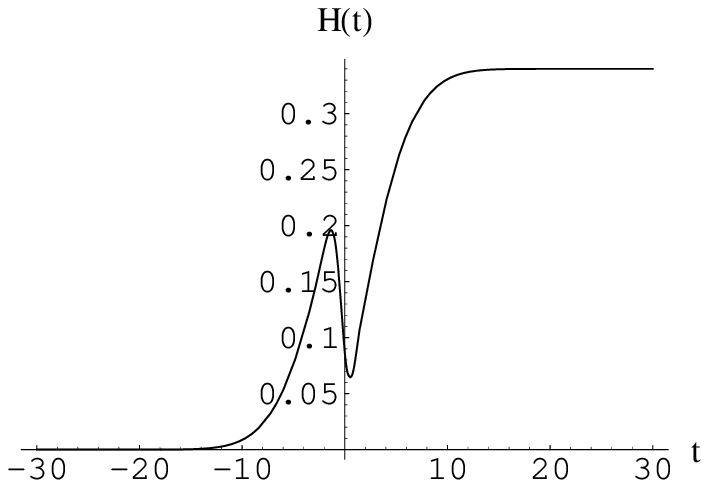}
\includegraphics[width=51mm]{H2.eps}
\includegraphics[width=51mm]{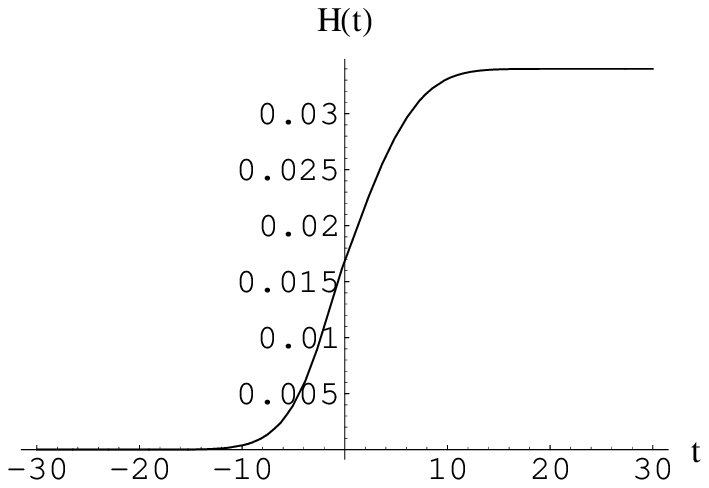}
\caption{Solutions of the scalar field (\ref{eom-gr}) and  Friedmann
equation (\ref{Fr-eq}) $\Phi$, $H$  for $m_p^2=0.1, 1, 10$ (left to right).
} \label{rolling-solution-ph-const0-diffmp}
\end{figure}

\subsection{Particular case of $p$-adic action}

The $p$-adic string model represents a popular toy-model that was
proposed in \cite{BFOW}. The model contains one scalar field and nonlocal
interaction. Formulated in arbitrary space-time dimensionality
the model has a parameter $p$, which is initially taken to be a
prime number. It was shown later the model describes quite well
some physical processes, despite its intrinsic limitations.
Recently this model started to attract attention as a
cosmological toy model for describing inflation \cite{p-adic-infl,
Lidsey, diff-infl}. The $p$-adic string model was later considered
as being minimally coupled to gravity in Friedmann-Robertson-Walker
metric. In this subsection we would like to show some intriguing
properties within this context.

The $p$-adic action is given by
\begin{equation}
S_p=\int d^d x L_p = \frac{1}{g_p^2} \int d^d x\left[ -\frac{1}{2}\phi
p^{-\frac{1}{2}\Box}\phi+\frac{1}{p+1}\phi^{p+1}
\right],~~~\frac{1}{g_p^2}=\frac{1}{g^2}\frac{p^2}{p-1}.
\label{p-adic action}
\end{equation}
The infinite number of space-time derivatives are manifest in the
pseudo-differential operator $p^{-\frac{1}{2}\Box}$, where
$\Box=-\pd^2+\nabla^2$, $p^{-\frac{1}{2}\Box}=e^{-\frac{1}{2}\ln p
\Box}$.

Considering action (\ref{p-adic action}) for $p=2$ after the
field redefinition $\varphi=e^{-\frac{1}{4}\ln p \Box}\phi$ we have
\begin{equation}
S_2=\int d^d x L_2 = \frac{1}{g_p^2} \int d^d x\left[
-\frac{1}{2}\varphi^2+\frac{1}{3}(e^{\frac{\ln 2}{4}
\Box}\varphi)^{3} \right] \label{p-adic action-p=2},
\end{equation}
which looks rather similar to the tachyon OSFT action in the level
truncation approximation without kinetic term with only
other difference in common signs.

Indeed, if we denote the approximate lagrangian for the OSFT case
$$
L_{OSFT_{approx}}=\frac{1}{2}\phi^2-\frac{1}{3}K^{3}(e^{k
\Box}\phi)^3
$$
we see that for the case of p-adic string model with $p=2$ if we
neglect the difference in the factors in exponential operator we
have
$$
L_{p=2}=-L_{OSFT_{approx}}.
$$

\begin{figure}
\centering
\includegraphics[width=51mm]{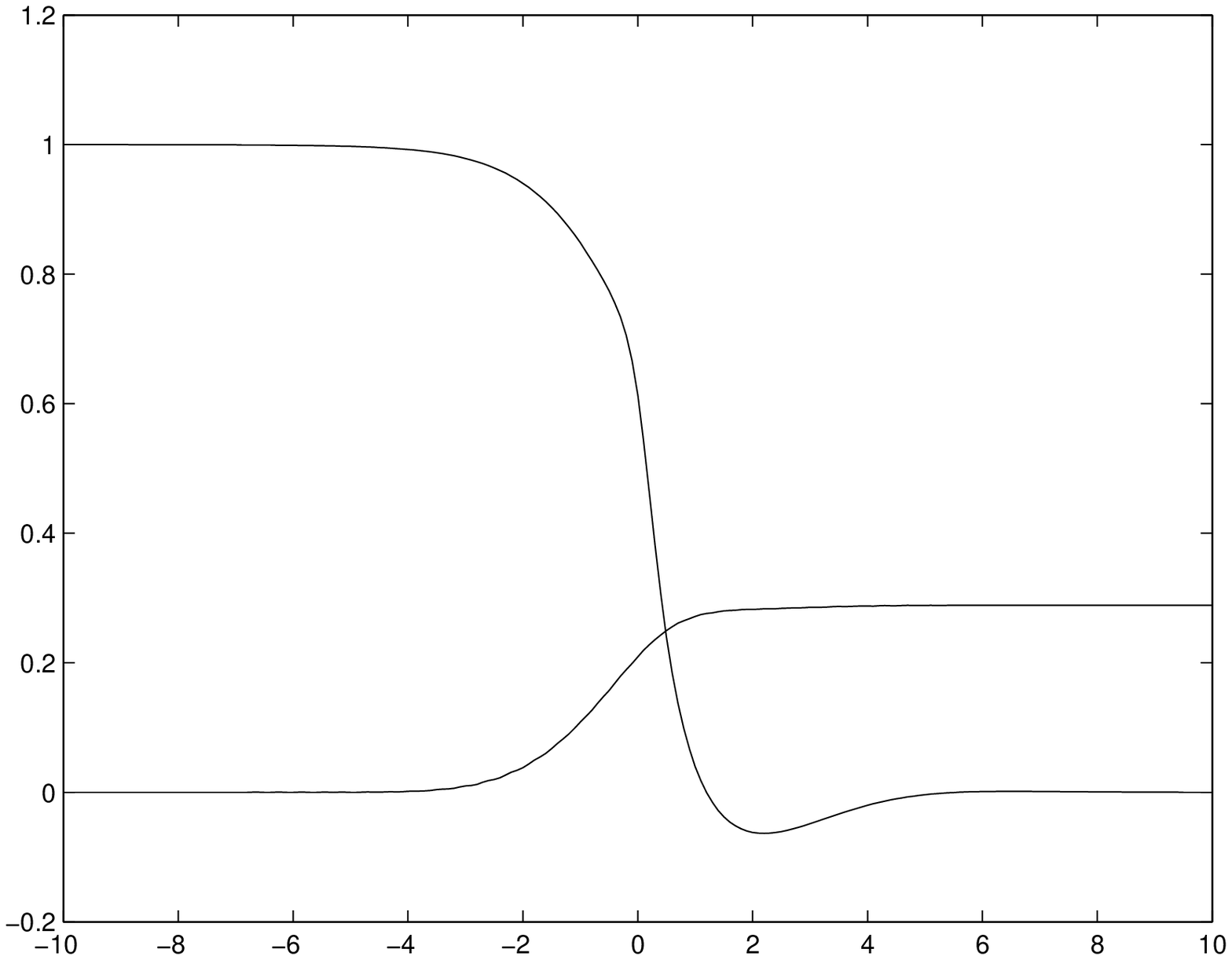}~~~~~~~~~~~~~~~
\includegraphics[width=51mm]{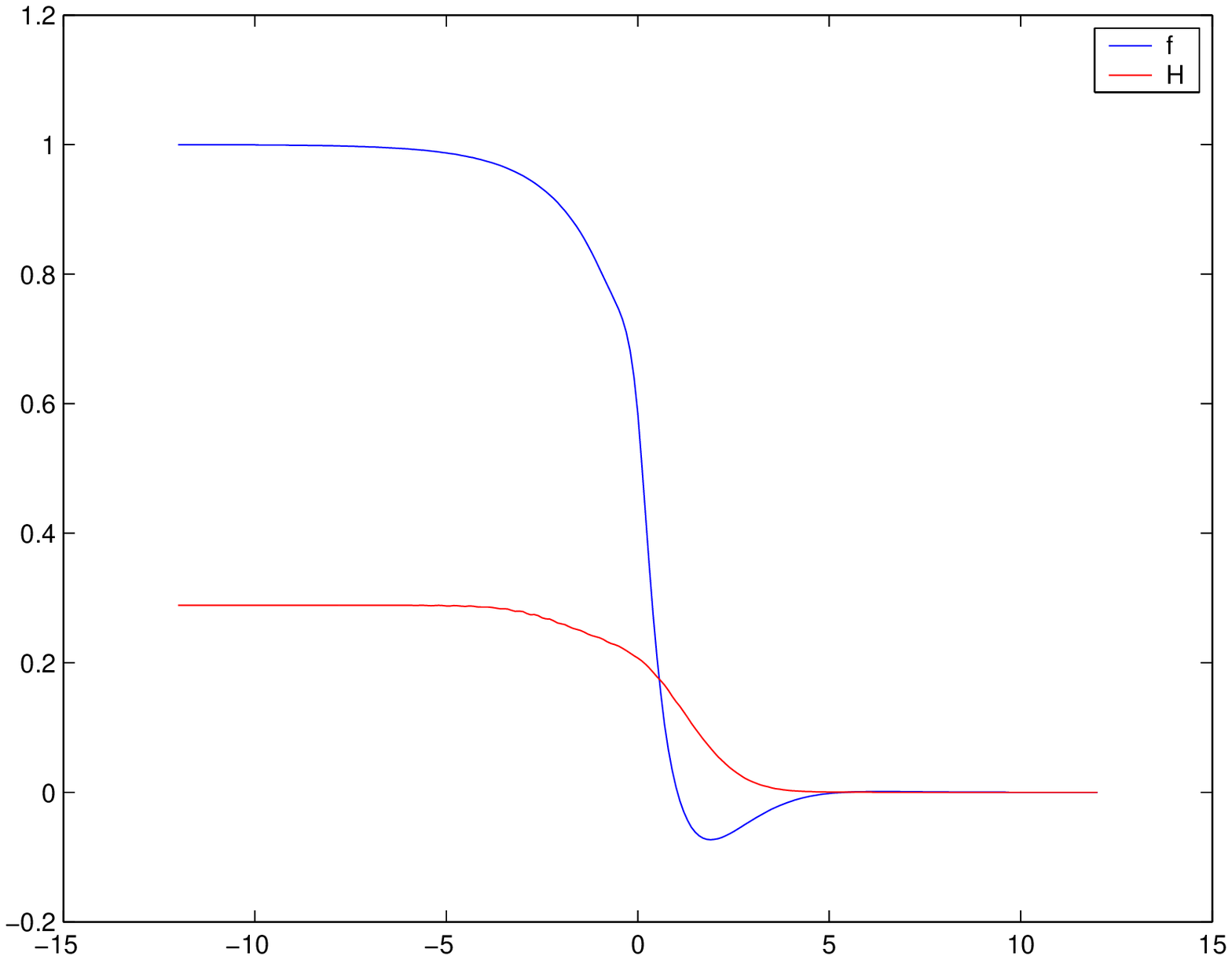}
\caption{Solutions of the Friedmann equations $\Phi$, $H$:
negative sign in front of the Lagrangian (left),
positive sign in front of the Lagrangian (right).} \label{pm}
\end{figure}

This observation does not affect dynamic in the case of
Minkowski space-time because sign in front of the Lagrangian does not
enter scalar field equation. The situation changes crucially in
Friedmann-Robertson-Walker background because the sign affects equation
for Hubble.
Thus if in the case of OSFT action in approximation neglecting
kinetic term we had monotonically increasing Hubble function
in the case of $p$-adic model for $p=2$ we obtain monotonically
decreasing Hubble function for the same scalar field configuration.
Corresponding solutions are presented on Fig.\ref{pm}.

On of the issues is that the $p$-adic sting model in the FRW case is not
currently known and the choosing the right sign is nontrivial,
especially due to the absence of usual canonical kinetic term.

\subsection{Energy and Pressure for Rolling Tachyon Solutions}
\begin{figure}
\centering
\includegraphics[width=51mm]{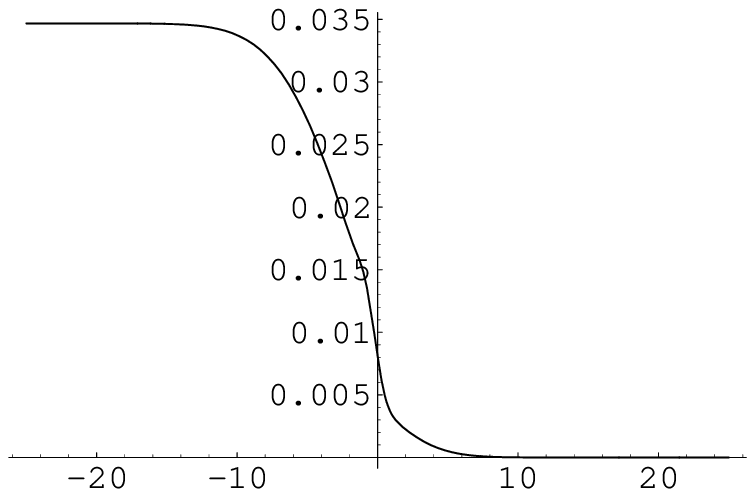}~~~~~~~~~~~~
\includegraphics[width=51mm]{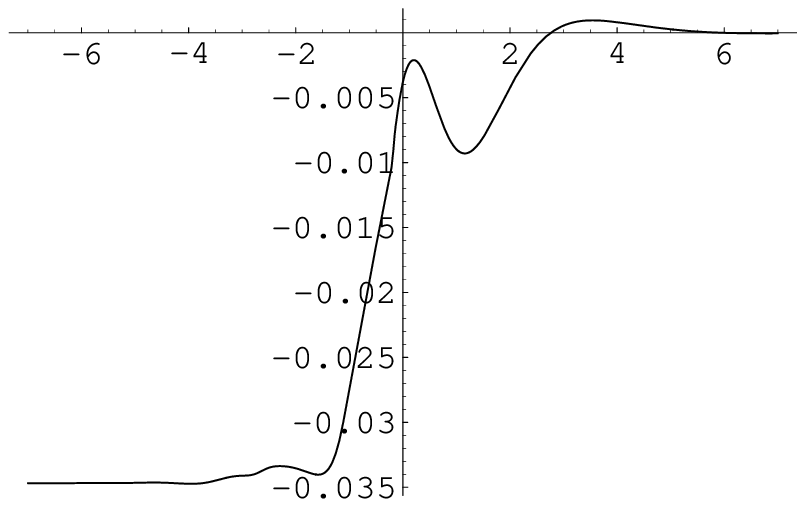}
\caption{Energy (left) and pressure (right) for solutions from Fig.\ref{rolling-solution-ph-0const}.}
\label{energy-pressure-0const}
\end{figure}

Usually when considering dynamics in nontrivial background such as FRW
it is instructive to consider simpler case of Minkowski metric.
The intriguing fact though is that there is no such a possibility since energy
conservation law forbids in Minkowksi case solutions like were presented in
Fig.\ref{rolling-solution-ph-0const}, \ref{rolling-solution-ph-const0}.
Nevertheless it is possible to qualitatively characterize new features of
the system.

Dynamics of energy and pressure for the scalar field from Fig.\ref{rolling-solution-ph-0const}
is presented on Fig.\ref{energy-pressure-0const}.
We can see that dynamics is different from what one might expect in Minkowski case.
We obtained solution for which both energy and pressure nontrivially tend to zero
at large times.
For the solution from Fig.\ref{rolling-solution-ph-const0}
the energy and pressure dynamics is different, see Fig.\ref{energy-pressure-const0}.
The pressure starts from zero and goes to a negative constant, while
the energy starts from zero and goes to a constant of the same value with
positive sign.

In both cases we obtained nontrivial partially negative pressure and
like in \cite{Gibbons} nontrivial equation of state what might be
especially interesting in the light of cosmological applications.

\begin{figure}
\centering
\includegraphics[width=51mm]{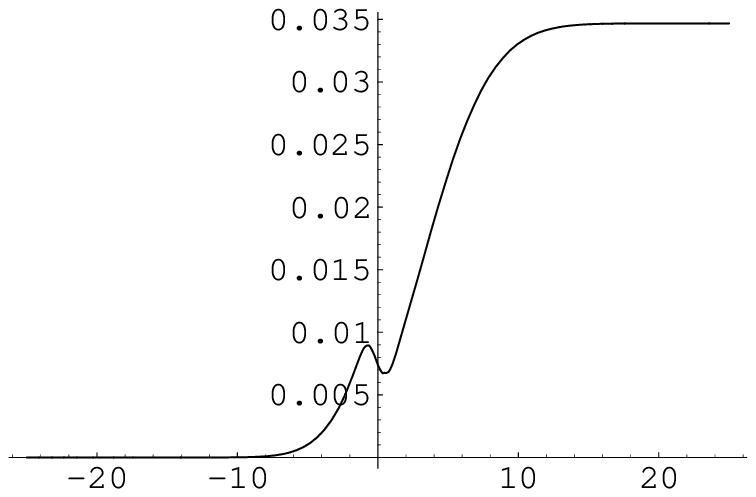}~~~~~~~~~~~~
\includegraphics[width=51mm]{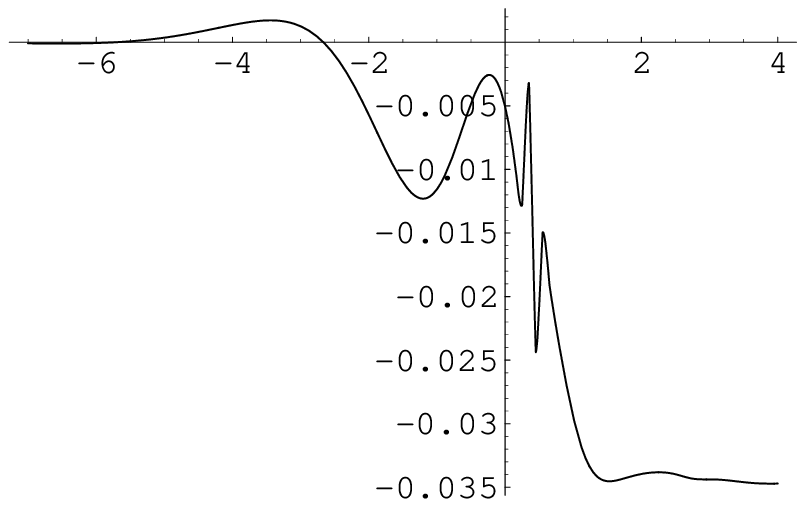}
\caption{Energy (left) and pressure (right) for solutions from Fig.\ref{rolling-solution-ph-const0}.}
\label{energy-pressure-const0}
\end{figure}

\subsection{Cosmological Applications}

To discuss cosmological consequences it is interesting to refer
to the concluding remarks and open questions of the work \cite{MZ}
which in fact initiated a whole series of investigation
of nonlocal dynamics in the models with infinitely many time
derivatives \cite{MZ}.
Original paper considered string field theory and $p$-adic theory
in Minkowski space time. It was noted that the most puzzling result
from physical point of view is that none of the
solutions obtained there appeared to represent tachyon matter.
In other words there were not found solutions where with
varying values of energy the pressure tends to zero for long
times. It is notable that solutions obtained in this work have
such behavior. As we can see on Fig.\ref{energy-pressure-0const}
the pressure goes to zero for long times at the end of the evolution.
In fact this represents first alternative of possible cosmological
applications of the obtained solutions.

The second alternative is to consider obtained solutions in the
context of description of very early Universe. This investigation is
motivated by the shape of the solutions with decreasing positive
Hubble function presented at the Fig.\ref{rolling-solution-ph-0const}.
The research in this direction was started in \cite{p-adic-infl} based
on p-adic string model  and followed by \cite{diff-infl} in which
CSSFT string model was considered as well.

The third alternative is to consider configuration from
Fig.\ref{rolling-solution-ph-const0} in the context of late time acceleration.

\newpage
\section{Summary}

In this paper we have studied tachyon dynamics in string theory in
Friedmann-Robertson-Walker background. The main results are:
\begin{itemize}
\item
It was shown that filed configurations which interpolate between vacua with
different scalar field energies are possible in the FRW background and
corresponding rolling solutions were obtained.

\item
We found nontrivial dynamics in the FRW background which differs from
dynamics in Minkowski case. Particularly, we found interpolating solution
which goes from maximum to the true minimum for which the pressure approaches
zero at long times while energy varies starting from the constant which is
equal to the brane tension and tends to zero at the end of the evolution.
\end{itemize}

Among main results we have also considered in details the
properties of constructed solutions and have shown that the
evolution in tachyon potential is possible in both directions,
and what is even more puzzling profiles for Hubble functions are
different in those cases.
In one direction Hubble function is an increasing almost monotonic
function, while in another we observe significant oscillations
during evolution. Moreover it was shown that by varying string scale
we can change the shape of Hubble function and for some
string/Planck scale ratio the oscillations disappear.
We also considered dynamics in $p$-adic string model for particular
value of a parameter $p=2$ and have shown that dynamics in the FRW
background drastically differs from the corresponding one in usual
string theory.

It is interesting to note that for numerical construction of solutions
presented here we had to abandon explicit iterative techniques which
were very successful in previous investigations of rolling tachyon solutions,
see \cite{Lulya-paper,diff-infl} and references therein.
Instead we used slower but more generic relaxation methods
with Crank-Nicholson scheme to compute nonlocal operators.

\section*{Acknowledgements}
The author would like to thank I. Aref'eva, R. Bradenberger, A.-C.
Davis, J. Khoury, N. Nunes, F. Quevedo, D. Seery, D. Wesley and
especially D. Mulryne and  Ya. Volovich for  useful discussions. The
author would also like to thank the Perimeter Institute for
hospitality while the part of this work was done.  The author
gratefully acknowledge the use of the UK National Supercomputer,
COSMOS, funded by PPARC, HEFCE and Silicon Graphics. This work is
supported by  the Centre for Theoretical Cosmology, in Cambridge.

\section*{Appendix}

In this appendix we will prove the following identity
\begin{equation}
\label{idd-app} \int\limits_0^1 d\rho (e^{ \rho
\partial^2 }\varphi ) \overleftrightarrow{\partial^2} (e^{ (1
- \rho)\partial^2 } \phi)= \varphi \overleftrightarrow
{e^{\partial^2 }}\phi,
\end{equation}
where symbol $e^{ \rho\partial^2 }\varphi$ comprehend as
\cite{Lulya-paper-2}

It is well known, that for functions $\Phi(t)$ which are continuous
and bounded on the real axis the following identity have a place
\begin{equation}
\label{lim} \lim_{\rho \to + 0} C_{\rho}[\Phi](t)=\Phi(t)
\end{equation}

We can formulate the following lemma

\begin{lemma}
For continuous and bounded functions $\psi(t)$  and $\varphi(t)$ the
following identity has a place
\begin{equation}
\label{lemma-app} \int\limits_0^1 d\rho C_{\rho}[\varphi](t)
\overleftrightarrow{\partial^2} C_{1-\rho}[\psi](t)= \varphi
\overleftrightarrow { C_{1}}\psi, ~~~\partial^2 \equiv
\frac{d^2}{dt^2}
\end{equation}
where left side we understand as
$$
\lim_{\epsilon_1\to +0} \lim_{\epsilon_2\to +0}
\int_{\epsilon_1}^{1-\epsilon_2} d\rho d\rho C_{\rho}[\varphi](t)
\overleftrightarrow{\partial^2} C_{1-\rho}[\psi](t)$$ and right hand
side we understand as
$$ \varphi
\overleftrightarrow { C_{1}}\psi= \varphi(t)
C_{1}[\psi](t)-C_{1}[\varphi](t) \psi(t)
$$
\end{lemma}
\begin{proof}
We have
$$
\int\limits_{\varepsilon_1}^{1-\varepsilon_2} d\rho
C_{\rho}[\varphi] \overleftrightarrow{\partial^2}
C_{1-\rho}[\psi]\int\limits_{\varepsilon_1}^{1-\varepsilon_2} d\rho
C_{\rho}[\varphi] (\partial^2C_{1-\rho}[\psi])-
\int\limits_{\varepsilon_1}^{1-\varepsilon_2}d\rho ( \partial^2
C_{\rho}[\varphi])2C_{1-\rho}[\psi]
$$
We will use the fact that for $\rho>0$ the function
$C_{\rho}[\varphi](t)$ is a solution of the diffusion equation, i.e
$$
\frac{\pd^2}{\pd t^2}  C_{\rho}[\varphi](t)=\frac{\pd}{\pd \rho}
C_{\rho}[\varphi](t) ,~~~~\rho>0
$$
$$
\frac{\pd^2}{\pd t^2} C_{1-\rho}[\varphi](t)=-\frac{\pd}{\pd \rho}
C_{1-\rho}[\varphi](t),~~~~\rho>0
$$
The proofs of these identities follow from integral representation.

Taking into account the identities written above we get
$$
\int\limits_{\varepsilon_1}^{1-\varepsilon_2}
 d\rho C_{\rho}[\varphi] \overleftrightarrow{\partial^2}C_{1-\rho}[\psi]\int\limits_{\varepsilon_1}^{1-\varepsilon_2}
 d\rho C_{\rho}[\varphi]
   (\partial^2 C_{1-\rho}[\psi])-
\int\limits_{\varepsilon_1}^{1-\varepsilon_2}
  d\rho (\partial^2 C_{\rho}[\varphi]) C_{1-\rho}[\psi]
$$
$$
=-\int\limits_{\varepsilon_1}^{1-\varepsilon_2}d\rho
  \left(C_{\rho}[\varphi ]\frac{\partial}{\partial \rho}C_{1-\rho}[ \psi]+
        \frac{\partial}{\partial \rho}C_{\rho}[\varphi]C_{1-\rho}[\psi]\right)
$$
$$
=-\int\limits_{\varepsilon_1}^{1-\varepsilon_2}d\rho
\frac{\partial}{\partial \rho} ( C_{\rho}[\varphi ]C_{1-\rho}[
\psi])= -C_{1-\varepsilon_2}[\varphi ]C_{\varepsilon_2}[
\psi]+C_{\varepsilon_1}[\varphi ]C_{1-\varepsilon_1}[ \psi]
$$
and (\ref{lim}), we can  take the limit
$\varepsilon_1\rightarrow+0$, $\varepsilon_2\rightarrow+0$ and get
(\ref{lemma-app}).
\end{proof}

\end{document}